\documentclass[aps,nofootinbib,superscriptaddress,notitlepage]{revtex4-1}
\usepackage{array}
\usepackage{graphicx}
\usepackage{amsmath}
\usepackage{amsthm}
\usepackage{amsfonts}
\usepackage{amssymb}
\usepackage{multirow}
\usepackage[english]{babel}
\usepackage[inline]{enumitem}
\usepackage{hyperref}
\usepackage{color}
\usepackage[dvipsnames]{xcolor}

\hypersetup{
	colorlinks,
	linkcolor={blue},
	citecolor={blue},
	urlcolor={blue}
}
\usepackage{cleveref}

\pdfoutput=1
\newtheorem{definition}{Definition}
\newtheorem{theorem}{Theorem}
\newtheorem{lemma}{Lemma}

\newtheorem{corollary}{Corollary}

\newcommand{\ket}[1]{\left|  #1 \right \rangle}

\newcommand{\ketbra}[2]{\left | #1 \right\rangle \left \langle  #2 \right |}

\newcommand{\set}[1]{\left\{ #1 \right\}}

\newcommand{\ProjP}[1]{\widetilde{\mathcal{P}}_{#1}}

\begin{document}

\graphicspath{{./figures/}}

\title{A Contextual $\psi$-Epistemic Model of the $n$-Qubit Stabilizer Formalism}
\date{\today}
\author{Piers Lillystone}
\affiliation{Institute for Quantum Computing and Department of Physics and Astronomy, University of Waterloo, Waterloo, Ontario N2L 3G1, Canada}
\author{Joseph Emerson}
\affiliation{Institute for Quantum Computing and Department of Applied Mathematics, University of Waterloo, Waterloo, Ontario N2L 3G1, Canada}
\affiliation{Canadian Institute for Advanced Research, Toronto, Ontario M5G 1Z8, Canada}

\begin{abstract}
Contextuality, a generalization of non-locality, has been proposed as \emph{the} resource that provides the computational speed-up for quantum computation. For universal quantum computation using qudits, of odd-prime dimension, contextuality has been shown to be a necessary and possibly sufficient resource ~\cite{Howard2014}. However the role of contextuality in quantum computation with qubits remains open. The $n$-qubit stabilizer formalism, which by itself cannot provide a quantum computer super-polynomial computational advantage over a classical counterpart ~\cite{Gottesman1998}, is contextual. Therefore contextuality cannot be identified as a sufficient resource for quantum computation. However it can be identified as a necessary resource ~\cite{Raussendorf2017}. In this paper we construct a contextual $\psi$-epistemic ontological model of the $n$-qubit stabilizer formalism, to investigate the contextuality present in the formalism. We demonstrate it is possible for such a model to be outcome deterministic, and therefore have a value assignment to all Pauli observables. In our model the value assignment must be updated after a measurement of any stabilizer observable. Remarkably, this includes the value assignments for observables that commute with the measurement. A stronger manifestation of contextuality than is required by known no-go theorems.
\end{abstract} 
\maketitle

\section{Introduction}

Contextuality is well-established as a fundamental non-classical feature of quantum theory~\cite{Bell1967, Kochen1967, Mermin1990, Peres1991, Mermin1993,
Simmons2017}, including as a special case Bell non-locality. This has naturally lead to it being proposed as a candidate for understanding the source of quantum computational speed-up. Recent work in classifying contextuality as a resource for quantum computation has proved a fruitful research field. Howard \emph{et al.}~\cite{Howard2014} proved that contextuality is a necessary and possibly sufficient resource\footnote{With sufficiency dependent on the conjecture~\cite{Veitch2014} that all $\rho \notin \mathcal{P}_{\text{SIM}}$, where $\mathcal{P}_{\text{SIM}}$ is the set of ancilla states not useful for magic state distillation routines, can be distilled to a magic state.} for universal quantum computation, in the context of magic-state injection schemes, with odd-prime dimension qudits. For qubit quantum computation contextuality has been proven to be a necessary resource~\cite{Delfosse2015,Bermejo-vega2017,Raussendorf2017}. Further it can even quantify the amount of computational advantage in magic-state injection schemes and measurement-based models of quantum computation~\cite{Anders2009,Raussendorf2001,Veitch2012,Raussendorf2013, Veitch2014, Bermejo-vega2017, Okay2017,Abramsky2017,Catani2017}.

However, despite recent progress in identifying contextuality as a resource for quantum computation, it can only ever be considered as a necessary resource. This is because the $n$-qubit stabilizer formalism, a  fundamental subset (which we call a subtheory) of quantum operations that is not universal for quantum computation~\cite{Gottesman1997,Bravyi2005,Gottesman2009}, admits an efficient classical simulation~\cite{Gottesman1998,Nest2010,Aaronson2004,Pashayan2017}, but is contextual under any standard definition of contextuality~\cite{Mermin1990,Peres1991,Mermin1993,Anders2009,Howard2013,Krishna2017,Lillystone2018}. Therefore contextuality cannot be classified as sufficient for universal quantum computation.

This motivates the question: can the concept of contextuality be re-characterized or refined such that a revised notion can be identified as both a necessary and sufficient resource for universal quantum computation? If such a redefinition is possible it would unify the foundational principle that contextuality is the fundamental non-classical property of quantum theory~\cite{Bell1967,Kochen1967,Mermin1993} which justifies the widely held belief that quantum computation is more powerful than classical computation~\cite{Preskill2012}. A natural  next step toward answering this question is to understand exactly how contextuality manifests in the $n$-qubit stabilizer formalism by constructing an explicit contextual ontological model of that subtheory.

Within the landscape of ontological models there are two primary types~\cite{Harrigan2007,Harrigan2010,Leifer2014}: $\psi$-ontic ontological models wherein a quantum state is uniquely specified by the physical state,  and $\psi$-epistemic ontological models wherein quantum states are not fully specified by the underlying physical state of system, i.e. there exist physical states that can be prepared by multiple quantum states. A $\psi$-epistemic ontological model provides for a classical-like model of quantum theory and is appealing because it provides a means to evade the measurement problem while maintaining a classically natural notion of objective reality. For example, the Kochen-Spekker model of a qubit ~\cite{Kochen1967}, the 8-state model of the qubit stabilizer formalism~\cite{Wallman2012c, Blasiak2013}, and Gaussian quantum mechanics~\cite{Bartlett2012} are all $\psi$-epistemic models of subsets of quantum theory. Conversely, in a $\psi$-ontic model the quantum state fully describes the state of a system and hence cannot give a classically natural account of quantum theory. It remains an open question whether a physically plausible $\psi$-epistemic model of quantum theory is possible, with recent work~\cite{Ruebeck2018} demonstrating that all previously known $\psi$-epistemic models of quantum theory cannot reproduce state update.

The Gottesman-Knill theorem can be used to define a $\psi$-ontic ontological model of the $n$-qubit stabilizer subtheory. This is also the case for the computational basis phase model~\cite{Dehaene2003,Nest2010}, which can be used to construct a weak-simulation of the $n$-qubit stabilizer subtheory and beyond~\cite{Bravyi2016}. For the, non-contextual, $n$-qudit stabilizer subtheory a $\psi$-epistemic model is possible~\cite{Spekkens2014,Catani2017}. However, for qubits the only previously known $\psi$-epistemic model was for the single qubit stabilizer formalism, the 8-state model~\cite{Wallman2012c}. Indeed, for a single qubit, all of quantum theory is non-contextual~\cite{Kochen1967}, at least if the state-update rule is not explicitly required in the model.

Hence it has been an interesting open question whether an $\psi$-epistemic ontological model of the full $n$-qubit stabilizer subtheory can be constructed, which is expected because we know that this subtheory admits an efficient classical description. In this paper we answer this open question in the positive by explicitly constructing a contextual $\psi$-epistemic model of the $n$-qubit stabilizer subtheory. We expect that this model will serve as a fundamental insight into the structure of contextuality within quantum theory and help clarify the role of contextuality as a resource for quantum speed-up for which the stabilizer subtheory serves as a stepping stone. 

As we will see, the $\psi$-epistemic ontological model constructed in this paper is outcome deterministic, that is, physical states of the system fully encode the outcome of any stabilizer measurement. Outcome determinism is often seen as a core motivating principle behind the traditional notion of contextuality, often referred to as Bell-Kochen-Specker contextuality~\cite{Bell1967,Kochen1967}. Therefore it of interest that the model presented here satisfies this principle, a feature shared by previous contextual explanations of the Mermin-Peres square~\cite{Kleinmann2011,Larsson2012}. However the model does not satisfy the requirement that the possible static value assignments respect all functional relationships between sets of commuting observables. Rather it only requires that a given value assignment is consistent with the current state and uses measurement update rules to satisfy the functional relationship between outcomes of commuting observables. We also provide a restructuring of the model in the appendix, which has the appealing feature of being an \emph{always}-$\psi$-epistemic model, but is no longer outcome deterministic.

In our model a correct update to the value assignments of the model requires the physical states to encode $n-1$ generators of a stabilizer state. However, we highlight as an open question whether a generic $\psi$-epistemic model of the $n$-qubit stabilizer formalism requires \lq\lq{}almost\rq\rq{} complete knowledge of stabilizer group of a state.

\section{Preliminaries}

\subsection[]{Ontological Models}

To build a model of the stabilizer formalism we use the ontological models formalism. To define an ontological model we begin with the mathematical framework of operational theories. An operational theory provides a framework for predicting the observed outcomes of experimental procedures, which can be decomposed into a sequence of preparations, transformations, and measurements. Hence an operational theory gives a specification of the experimentally observed statistics $\text{Pr}(k|P,T,M)$ ~\cite{Harrigan2010}.

An ontological model provides a description of a set of experimental statistics by supposing that there exists a well defined notion of a real physical state of a system, referred to as the ontic state $\lambda$, which encodes all physically knowable properties of the system. The ontological models formalism then describes preparations, transformations, and measurements in terms of probability densities, stochastic processes, and response functions over the set of ontic states, $\Lambda$, which we call the ontic space:

\begin{tabular}{p{2.65cm}p{12cm}}
\emph{Preparations:} & A preparation $P$ is represented in an ontological model by a probability density $\mu_P$ over the ontic space:
$\mu_P: \Lambda \rightarrow [0,1]$, where $\int_\Lambda \mu_P(\lambda)d\lambda = 1$. \\
\emph{Transformations:} & A transformation $T$ is represented in an ontological model by a stochastic map $\Gamma_T$ between ontic states: $\Gamma_T: \Lambda \times \Lambda \rightarrow [0,1]$, where $\int_\Lambda\Gamma_T(\lambda^\prime, \lambda)d\lambda^\prime =1, \, \forall \lambda \in \Lambda$. \\
\emph{Measurements:} & A measurement of a POVM, $M = \set{E_k}$, is represented in an ontological model by a set of conditional probability distributions over outcomes of the POVM $\set{\xi_M(k|\lambda)| \, E_k \in M}$ , where $\xi_M(k): \Lambda \rightarrow [0,1]$, and $\sum_k \xi_M(k|\lambda) = 1, \, \forall \lambda \in \Lambda$.
\end{tabular}
\\

To reproduce the statistics of an operational theory $\text{Pr}(k|P,T,M)$ the ontological model must satisfy:
\begin{align}
\text{Pr}(k|P,T,M) = \int_{\Lambda^\prime}\int_\Lambda \xi_M(k|\lambda^\prime)\Gamma_T(\lambda^\prime,\lambda)\mu_P(\lambda) d\lambda d\lambda^\prime.
\end{align}

For example in this paper we are interested in the operational theory defined by the $n$-qubit stabilizer formalism. I.e. convex mixtures of; $n$-qubit pure stabilizer state preparations, unitary transformations from the Clifford group, and projective measurements of Pauli observables\footnote{In the fragment notation the qubit stabilizer subtheory is defined as $\mathcal{F}_{stab, 2^n} =<\mathcal{S}(\mathcal{D}(\mathcal{H}_{2^n})), {Cl}_{2^n}, \widetilde{\mathcal{P}}_n >$~\cite{Leifer2014}.}.

\subsubsection[]{$\psi$-Epistemic models}

The ontological models formalism provides a platform for investigating statements about how physical reality may be structured. A central question in any such investigation is whether a quantum state provides a unique description of physical reality, i.e. whether every ontic state $\lambda$ can be associated to a single quantum state $\rho$. Ontological models where this is the case are termed $\psi$-ontic. Conversely, ontological models where the quantum state does not give a full description of physical reality, but represents some lack of knowledge of the true state of affairs, are termed $\psi$-epistemic:

\begin{definition}\label{def:psiontic}
A \textbf{$\psi$-ontic} ontological model, $\mathcal{O}$, is one where all pure quantum states' probability distributions, $\mathcal{P}(\mathcal{O})$, are non-overlapping:
\begin{align*}
\int_\Lambda \mu_{P_\psi} (\lambda) \mu_{Q_\phi} (\lambda) d\lambda = 0, \, \forall P_\psi, Q_\phi \in \mathcal{P}(\mathcal{O}).
\end{align*}
Otherwise it is \textbf{$\psi$-epistemic}.
\end{definition}

For example, in the ontological model defined by the Gottesman-Knill theorem every possible tableau of generators is identified as an ontic state. Hence each ontic state corresponds to exactly one stabilizer state, making the model $\psi$-ontic. Note, a stabilizer state can be represented by $2^{n(n-1)/2}\prod_{k=0}^{n-1} (2^{n-k}-1)$ tableaus ~\cite{Aaronson2004}, therefore there are as many ontic states associated to a single stabilizer state.

An alternative way to express the criteria for an ontological model to be $\psi$-epistemic is to notice that \cref{def:psiontic} requires that there exists at least one ontic state $\lambda$ that is in the support of more than one quantum state's support:
\begin{definition}\label{def:psiepistemic}
An ontological model is \textbf{$\psi$-epistemic} if there exists an ontic state $\lambda$ that is in the support of more than one pure quantum state, i.e. mathematical we require:
\begin{align*}
\exists \lambda \in \text{supp}(\mu_{P_\psi}) \cap \text{supp}(\mu_{Q_\phi}), \, P_\psi, Q_\phi \in \mathcal{P}(\mathcal{O}),
\end{align*}
where the support is defined as $\text{supp}(\mu_{P_\psi}) = \set{\lambda|\mu_{P_\psi}(\lambda) > 0, \lambda \in \Lambda}$.
\end{definition}

Requiring that there only exists one ontic state in the support of more than one quantum state is a very weak condition. A more natural requirement of an ontological model is that all ontic states are an element of more than one quantum state's support\footnote{Such a requirement firmly places the role of the quantum state at the level of a state of knowledge, i.e. given an ontic state we cannot determine which quantum state that ontic state.}, such a model is called an \text{always-$\psi$-epistemic} model;

\begin{definition}\label{def:alwayspsiepi}
An ontological model is an \textbf{always}-$\psi$-\textbf{epistemic} if the model satisfies; 

$\forall \lambda \in \Lambda$ there exist two operationally distinct pure state preparations $P_\psi, Q_\phi \in \mathcal{P}(\mathcal{O})$ such that $\lambda \in \text{supp}(\mu_{P_\psi}) \cap \text{supp}(\mu_{Q_\phi})$.
\end{definition}

Additionally we can require that all ontic states are in the support of the same number of quantum states, we term ontological models that satisfy this \textbf{symmetric}-always-$\psi$-epistemic models.

One of the interesting features of always-$\psi$-epistemic models is their relation to weak simulations ~\cite{Nest2010}. Ideally in a weak simulation the act of sampling should not give us enough information to calculate the final distribution we wish to sample an outcome from. Similarly in an always-$\psi$-epistemic model having knowledge of which ontic state was prepared is insufficient to determine the quantum state of the system and therefore the full experimental statistics. However both are guaranteed to agree with the experimental statistics over many repeated runs, and never disagree with the experimental certainties.

\subsection[]{Traditional Contextuality}

Traditional non-contextuality, originally proposed by Kochen-Specker and Bell ~\cite{Bell1967, Kochen1967, Mermin1990, Peres1991, Mermin1993}, is defined by assuming it is possible to assign values to sets of commuting observables, such that functional relationships between observables are satisfied. So given an observable $M$ and an ontic state $\lambda$, $\lambda$ deterministically specifies the outcome of a measurement $M$. Mathematically we express this value assignment as $v_\lambda (M) \in \set{E_k}$, where $\set{E_k}$ is the set of eigenvalues of $M$ indexed by $k$.

Within quantum theory if an observable $M$ is simultaneously measurable with some set of mutually commuting observables $\set{A,B,..}$ then any functional relationship between these observables, $f(M,A,B,..)=0$, will also be satisfied by the simultaneous eigenvalues of these observables, so $f(E_k,E_a,E_b,..) =0$. The assumption of traditional non-contextuality then assumes that this functional relationship should also hold for any value assignment, i.e. $f(v_\lambda (M) ,v_\lambda (A) ,v_\lambda (B),...)=0, \forall \lambda \in \Lambda$.

The above assumption implicitly assumes that the value assignments should satisfy the functional relationship for \emph{any} set of commuting observables, i.e. $\set{M,A,B,...}$ and $\set{M,L,M,...}$, even if such sets do not mutually commute. It is this assumption that the model presented here explicitly breaks from. We only require that the value assignment given by an ontic state is consistent with the quantum state and not hypothetical commuting sets of measurements. The functional relationships between commuting observables are then satisfied via a measurement update rule, ensuring the post-measurement value assignment is consistent with the post-measurement state.

\subsection[]{The $n$-Qubit Stabilizer Formalism}

As we are interested in building a contextual ontological model of stabilizer circuits we restrict ourselves to the unitary/PVM $n$-qubit stabilizer subtheory. 
The operational theory we are interested in is therefore $p(k| \, \rho, \, C, M)$, where $\rho \in \mathcal{S}(\mathcal{H}), \, C \in {Cl}_{2^n}, \, M \in \widetilde{\mathcal{P}}_n$.

The most basic element of the stabilizer formalism is projective Pauli group $\widetilde{\mathcal{P}}_n$, which we define as;

\begin{definition}\label{def:ProjP}
The \textbf{projective Pauli group} is the standard Pauli group, $\mathcal{P}_n$, modulo global phases;
\begin{align*}
\widetilde{\mathcal{P}}_n = \set{\bigotimes_i^n P_i | \, P_i \in \set{\mathbb{I},X,Y,Z}} = \mathcal{P}_n / U(1).
\end{align*}
\end{definition}

We say an operator \emph{stabilizes} a state if applying that operator to the state leaves it unchanged, i.e. it is an element of the $+1$ eigenspace of the operator. In the stabilizer formalism we choose the projective Pauli group with $\pm 1$ phases as our set of \emph{stabilizer operators}. Hence the set of stabilizer operators is isomorphic to $ (\mathbb{Z}_2 \times \widetilde{\mathcal{P}}_n) \backslash -\mathbb{I}\cong\mathcal{S}_n $, where we have removed $-\mathbb{I}$ as it stabilizes no state.
Note the full Pauli group, $\mathcal{P}_n$, contains operators that square to negative identity and therefore do not stabilize any state, i.e. non-Hermitian Pauli operators such as $iX\in \mathcal{P}_n$.

For the purposes of the model we are interested in pure stabilizer states. Mixed stabilizer states in the model are considered to be convex combinations of pure stabilizer states. The set of pure stabilizer states is defined as the set of pure states that are in the joint $+1$ eigenspace of $n$ mutually commuting stabilizer operators, i.e. $\mathcal{S}(\psi) = \set{S| \, S\ket{\psi} =\ket{\psi}, \, S \in \mathcal{S}_n}$ where $[S,S^\prime]=0, \, \forall S,S^\prime \in \mathcal{S}(\psi)$. From this it can be shown that the \emph{stabilizer group} of $\ket{\psi}$, $\mathcal{S}(\psi)$, is a maximal abelian subgroup of the Pauli group ~\cite{Gottesman2009}, with $n$ generators. Further the stabilizer group uniquely specifies the stabilizer state. To see this consider that if we define $\mathcal{S}(\psi) = \left<g_1,g_2,...,g_n\right>, \, g_i \in \mathcal{S}_n$ then;
\begin{align}\label{eqn:StabGroupRep}
\rho_\psi = \prod_{i=1}^n \frac{\mathbb{I} + g_i}{2} = \frac{1}{2^n}\sum_{S\in \mathcal{S}(\psi)} S,
\end{align} 
i.e. we repeatedly project onto each generator's $+1$-eigen-subspace until we reach a $1$-dimensional subspace.

The unitaries in the stabilizer formalism are unitaries from the Clifford group, $Cl_{2^n}$, which is defined as the normalizer of the Pauli group;
\begin{align}\label{eqn:CliffGroup}
Cl_{2^n} = \set{U \in U(2^n) |\, U S U^\dagger \in \mathcal{S}_n, \, \forall S \in \mathcal{S}_n}.
\end{align}
Therefore Clifford gates, such as the Hadamard, phase, and CNOT gates, strictly map stabilizer states to stabilizer states.

As previous stated the observables measurable in the $n$-qubit stabilizer subtheory are all projective Pauli operators, where other stabilizer operators can be measured via a classical post-processing bit flip. Therefore the operational statistics any ontological model of the $n$-qubit stabilizer formalism must reproduce can expressed as;

\begin{lemma}\label{lemma:StabStats}
Given a stabilizer state $\ket{\psi} \in \mathcal{S}(\mathcal{H})$ and a Pauli observable $M \in \widetilde{\mathcal{P}}_n$ the expectation value of the observable is given by;
\begin{align}\label{eqn:StabStats}
{\left< M\right>}_\psi=\text{Tr}(M\rho_\psi) = \left\{ \begin{array}{rcc} 1: & M\in\mathcal{S}(\psi)& (C1), \\
 0: & \pm M\notin\mathcal{S}(\psi) & (C2),\\
 -1: & -M\in\mathcal{S}(\psi) & (C3),
  \end{array} \right.
\end{align}
Where $\mathcal{S}(\psi)$ denotes the stabilizer group of $\ket{\psi}$.
\end{lemma}
\begin{proof}
Simply substituting in equation \ref{eqn:StabGroupRep} to the LHS and using the trace-orthonormality of the projective Pauli operators, $\text{Tr}(P_i P_j) = 2^n \delta_{i,j}$, retreives the desired result.
\end{proof}

Note here we have used the expectation value for clarity over the corresponding projectors. 
As the expectation is bounded between $-1$ and $+1$, and the observables $M\in \widetilde{\mathcal{P}}_n$ have $\pm1$ eigenvalues, we can see that if $M\in\mathcal{S}(\rho)$ then the outcome must be 1 (C1), similarly -1 for $-M\in\mathcal{S}(\rho)$ (C3).
Condition 2 (C2) implies the outcome of the measurement is uniformly random, over $\set{+1,-1}$, if $\pm M \notin \mathcal{S}(\rho)$.

\subsection[]{The Stabilizer Formalism and Measurement Update Rules}

Central to the model's construction are the ontological measurement update rules. Here we recap how measurement updates are performed in the stabilizer formalism.

As stabilizer measurements are rank $2^{n-1}$ PVMs the update rules for them are relatively straight forward to state. Suppose we measured the stabilizer state $\ket{\psi}$ with Pauli observable $M$ and observed outcome $k\in\{0,1\}$. Then the post measurement state is given by (with appropriate renormalization);
\begin{align}
\rho_{\psi^\prime} &\propto \frac{\mathbb{I} + {(-1)}^k M}{2} \rho_\psi \frac{\mathbb{I} + {(-1)}^k M}{2} ,\nonumber\\
&\propto \sum_{S \in \mathcal{S}(\psi)} S + (-1)^k (MS + SM) + MSM ,\nonumber\\
& \propto  \sum_{S \in \mathcal{S}(\psi) | [S,M] = 0} S + (-1)^k MS ,\nonumber\\
&= \sum_{b \in \mathcal{B}(\mathcal{S}(\psi)) | [b,m] =0} {(-1)}^{\gamma(b)} P_b + {(-1)}^{\gamma(b) + k +\beta(b,m)} P_{b+m},
\end{align}

where we have used the definition $P_a P_b = {(-1)}^{\beta(a,b)}P_{a+b}\, | \, [P_a,P_b]=0$. Expressing the update rule this way allows us to infer that the post-measurement stabilizer group is constructed by removing all elements of $\mathcal{S}(\psi)$ that do not commute with $M$ and replacing them with compositions of $M$ and the commuting group elements. This allows us to even further simplify the the update rule;

Given any stabilizer group of a pure state $\mathcal{S}(\psi)$ and any projective Pauli operator $M$, it is always possible to express the group as $\mathcal{S}(\psi) =\left<G,h\right>$ where $G$ is a proper subgroup of $\mathcal{S}(\psi)$ such that $[G_i , M] = 0, \, \forall G_i \in G$ and $\pm M \notin G$. Therefore the post-measurement group, after a measurement of $M$ with outcome $k$, is homomorphic to $\mathcal{S}(\psi^\prime)= \left<G,(-1)^k M\right>$, i.e. either $h \neq (-1)^k M$ and $[h,M] \neq 0$ and therefore we remove $h$ and add $(-1)^k M$ in its place or $h = (-1)^k M$ and the pre-and-post measurement stabilizer groups are the same.

Many of known no-go theorems can be performed in the stabilizer formalism by using sequential measurement procedures. For example, the Mermin-Peres square, a proof of contextuality, and the PBR theorem, which can be implemented via adaptive stabilizer measurements.

\section{A simple model via a global value assignment}

Before building a full ontological model of the $n$-qubit stabilizer formalism, let us construct a model of a simpler subtheory. Namely the prepare-measure-discard $n$-qubit stabilizer subtheory, i.e. we allow preparations of pure stabilizer states and a terminating measurement of a Pauli observable. From the stabilizer statistics, lemma \ref{lemma:StabStats}, the outcome of a measurement, $M$, only requires knowledge of whether $\pm M \in \mathcal{S}(\rho)$, which we can simply encode in a value assignment over all Pauli observables:

\begin{definition}\label{def:ValAss}
The \textbf{value assignment} $\nu$ is a function that assigns $\set{+1,-1}$ to every element of the projective Pauli group $\widetilde{\mathcal{P}}_n$;
\begin{align}
\nu: \widetilde{\mathcal{P}}_n \longrightarrow \set{+1,-1},
\end{align}
such that $\nu(\mathbb{I}) = +1$.
\end{definition}

However for the majority of the this paper it will be more useful to use a binary representation of the value assignment, that is we map $+1 \rightarrow 0$ and $-1 \rightarrow 1$, which we term the \emph{phase function}\footnote{Similar to the value assignments used in Wigner functions ~\cite{Delfosse2015,Bermejo-vega2017,Raussendorf2017}};

\begin{definition}\label{def:PhaseFn}
The \textbf{phase function} $\gamma$ is a function that assigns $\set{0,1}$ to every element of projective Pauli group $\widetilde{\mathcal{P}}_n$;
\begin{align}
\gamma: \widetilde{\mathcal{P}}_n \longrightarrow \mathbb{Z}_2,
\end{align}
such that $\gamma(\mathbb{I}) = 0$.
\end{definition}
To move between the two representations we can use $\nu(M) = (-1)^{\gamma(M)}$ and $\gamma(M) = \frac{1}{2}(1 - \nu(M))$.
We denote the set of $n$-qubit phase functions as $\gamma_n = \set{\gamma| \, \gamma: \widetilde{\mathcal{P}}_n \longrightarrow \mathbb{Z}_2}$. 
 
This definition naturally introduces the notion of consistency between stabilizer states and assignments $\gamma$;

\begin{definition}\label{def:Consistent}
A phase function  $\gamma$ is \textbf{consistent} with a stabilizer state $\rho\in\mathcal{S}(\mathcal{H})$, denoted $\gamma \cong \rho$, with stabilizer group $\mathcal{S}(\rho) = \set{(-1)^{p_b} P_b | p_b \in \mathbb{Z}_2, \, P_b \in \widetilde{\mathcal{P}}_n}$ iff;
\begin{align*}
\gamma(b) = p_b, \, \forall b \in \mathcal{B}(\mathcal{S}(\rho)),
\end{align*}
where $\mathcal{B}$ denotes the binary sympletic representation of $\mathcal{S}(\rho)$.
\end{definition}

For the rest of this paper we will let lower case letters imply some index over the projective Pauli group such that they form a group homomorphism, for example the binary-sympletic representation.

Note in definition \ref{def:Consistent} the phase function retains consistency regardless of its value on points not in the stabilizer group. 
From this we can reproduce the stabilizer statistics by defining an epistemic state to be a uniform distribution over consistent phase functions and measurements to read-out the value of the phase function. Therefore the outcome of a measurement will be random if $\pm M \notin \mathcal{S}(\rho)$ and equal to the phase on $M$ if $\pm M \in \mathcal{S}(\rho)$. 

While the scope of the model is a highly restricted, what we can learn is that the phase function is sufficient to output the outcomes of measurements in the prepare and measure setting. The subsequently presented $\psi$-epistemic model of the $n$-qubit stabilizer formalism leverages this by effectively encoding the outcome of a measurement in the phase function. Then uses additional ontology to ensure that the post-measurement phase function is consistent with the post-measurement state, with randomization where necessary. Hence reproducing the quantum statistics.

\section{A contextual $\psi$-Epistemic Model of the $2$-Qubit Stabilizer Formalism}

Before building the full $n$-qubit model it will be insightful to present the $2$-qubit case, which has a slightly less abstract construction. Clearly to reproduce the outcome statistics of any prepare-transform-measure stabilizer circuit it is sufficient to only use the phase function as described earlier. However to handle sequential measurements the update rules requires additional ontology. 

The additional ontology we require is non-trivial strict subgroups of the stabilizer group of the state. Which for two qubits is just the group containing one projective Pauli operator, $G = \set{\mathbb{I},P| P\in\ProjP{2}}$. We are only interested in projective Pauli operators as the phase function contains the information on whether $P$ or $-P$ is an element of the stabilizer group.

This means the ontology of the $2$-qubit model can be expressed as $\Lambda = \left(\ProjP{2} \backslash \mathbb{I} \right)\times \gamma_2$. The uniform distributions, $\mu_\psi$, representing a pure stabilizer state $\psi \in \mathcal{S}(\mathcal{H}_{2^2})$ have support;
\begin{align}
\text{supp}(\mu_\psi ) = \set{\lambda = (P,\gamma)| P\in \widetilde{\mathcal{S}}(\psi)\backslash \mathbb{I}, \gamma \cong \psi},
\end{align}
where the uniformity of the distribution ensures a random outcome for observables not in the stabilizer group. The model is $\psi$-epistemic as some non-orthogonal states have mutual support. For example, the computation basis state $\ket{00}$ has support on all $\lambda = (P,\gamma)$ such that $\gamma(\mathbb{I}Z)=\gamma(Z\mathbb{I})=\gamma(ZZ) = 0$ and $P \in \set{\mathbb{I}Z,Z\mathbb{I},ZZ}$. Therefore the all-zero state shares ontic states with a state $\ket{\psi}$ if $\exists P^\prime \in \mathcal{S}(\psi)$ such that $P^\prime \in \set{\mathbb{I}Z,Z\mathbb{I},ZZ}$, and there exist phase functions consistent with both $\ket{00}$ and $\ket{\psi}$, i.e. $\ket{0+}$.

A Clifford operation's map on the ontology can be derived by considering it's map on the stabilizer operators, it's \emph{stabilizer relations}. The stabilizer relations for a given Clifford operation $C$ can be expressed as $C P_a C^\dagger = {(-1)}^{\gamma_c(a)}P_{c(a)}$ where $P_a, P_{c(a)} \in \ProjP{n}$. Using this expression we represent a Clifford operation in the model by a permutation of the ontic states;
\begin{align}
\Gamma_C: (P_a, \gamma) \longrightarrow (P_{c(a)}, \gamma^R + \gamma_c),
\end{align}
where $\gamma_c$ is the vector of the phases $\gamma_c(a)$ and $\gamma$ is reordered to match the Clifford permutation, therefore $(\gamma^R + \gamma_c)(a) = \gamma(c^{-1}(a)) + \gamma_c (a)$.

A measurement $M$'s response functions \lq\lq{}read-out\rq\rq{} the value of the phase function;
\begin{align}
\xi_{k, M} (\lambda) = \begin{cases}
1 & \text{if} \, \, \gamma(M) = k, \\
0 & \text{otherwise.}
\end{cases}.
\end{align}
Therefore the model is outcome deterministic and the construction clearly satisfies the stabilizer outcome statistics for single shot experiments, as we effectively have reproduced the prepare-measure model given in the previous section.

Using the stabilizer measurement update rules described previously we can define the stochastic measurement update map to be $\Gamma_{k|M}$, letting $k = \gamma(M)$;
\begin{align}\label{eqn:2qUpdate}
\Gamma_{k|M}: (P,\gamma) \mapsto (P^\prime, \gamma^\prime)
\begin{cases}
P\mapsto P^\prime = M & \\
\text{if} \, [M,P] \neq 0:\left.\begin{array}{l}\gamma^\prime(S) = \gamma(S) \\ \gamma^\prime(SM) = \gamma(S) + k + \beta(S,M) \end{array}\right\} & \forall S\in \ProjP{2}, \, | \,  [S,M]=[S,P] = 0 \\
\text{if} \, [M,P]=0: \, \,\gamma(PM) = \gamma(P) + k + \beta(P,M)  \\ 
\gamma^\prime(S) = \set{0,1} \, \text{w. eq. prob} & \forall S\in \ProjP{2}, \, | \, [S,M] \neq 0, \, [S,P] =0 \\
\gamma^\prime(S) = \gamma(S) & \text{Otherwise}
\end{cases}
\end{align}

\begin{lemma}\label{lemma:2qModel}
The ontological model presented above is a $\psi$-epistemic model of the $2$-qubit stabilizer subtheory.
\end{lemma}

\begin{proof}

The proof of correctness is subsumed by the general $n$-qubit proof, theorem \ref{theorem:Main}. The direct proof of this lemma has been included in the appendix. We have included the proof as it gives an intuitive understanding of the structure of update rules in the model.

\end{proof}

We have also provided a mathematica code base\footnote{Accessable on \href{https://github.com/PiersLillystone/FiniteStateMachine2QStabSubtheory}{github}, date last accessed April 7th 2019.} to simulate the $2$-qubit model.

\subsection[]{Features of the $2$-qubit Model}

\subsubsection[]{The Mermin-Peres Square}

The Mermin-Peres (MP) square~\cite{Mermin1990,Peres1991,Howard2013} is a well known proof of state-independent traditional contextuality. The proof only requires stabilizer measurement, therefore the $n$-qubit stabilizer formalism, for $n>1$, is traditionally contextual. There are several variations of the MP square for $2$-qubits. Here we will focus on demonstrating how the model reproduces the quantum statistics for the most common form of the square, figure \ref{table:MPsq}.

\begin{figure}[h!]
\centering
\begin{tabular}{m{0.75cm} m{0.5cm} m{0.75cm} m{0.5cm} m{0.75cm} m{0.5cm} m{1cm}}
$X\mathbb{I}$ & --- & $\mathbb{I}X$ & --- & $XX$  & $\rightarrow$ & $\mathbb{I}$ \\
$\,\,|$                 &     & $\,\,|$                 &      & $\,\,|$ &                    &                     \\
$\mathbb{I}Z$ & --- & $Z\mathbb{I}$ & --- & $ZZ$  & $\rightarrow$ & $\mathbb{I}$ \\
$\,\,|$                 &     & $\,\,|$                 &      & $\,\,|$ &                    &                     \\
$XZ$              & --- & $ZX$               & --- & $YY$  & $\rightarrow$ & $\mathbb{I}$ \\
$\,\,\downarrow$ &     & $\,\,\downarrow$   &      & $\,\,\downarrow$ &     &                     \\
$\,\,\mathbb{I}$ &     & $\,\,\mathbb{I}$   &      & $-\mathbb{I}$ &     &                     \\
\end{tabular}
\caption{A Mermin-Peres square for $2$-qubits, with tensor notation suppressed. Each row and column constitutes a set of commuting observables, who's product is given at the end of the arrows.}
\label{table:MPsq}
\end{figure}

Under the assumptions of traditional non-contextuality there is no value assignment that satisfies the functional relationships between all commuting observables in the Mermin-Peres square. i.e. $\nexists v_\lambda \in \set{+1,-1}^{|\ProjP{2}|}$ such that;
\begin{align*}
v_\lambda(X\mathbb{I})v_\lambda(\mathbb{I}X)v_\lambda(XX) &= +1, \\
v_\lambda(Z\mathbb{I})v_\lambda(\mathbb{I}Z)v_\lambda(ZZ) &= +1, \\
v_\lambda(X\mathbb{I})v_\lambda(\mathbb{I}Z)v_\lambda(XZ) &= +1, \\
v_\lambda(Z\mathbb{I})v_\lambda(\mathbb{I}X)v_\lambda(ZX) &= +1, \\
v_\lambda(XZ)v_\lambda(ZX)v_\lambda(YY) &= +1, \\
v_\lambda(XX)v_\lambda(ZZ)v_\lambda(YY) &= -1.
\end{align*}

The model here evades such a contradiction by implicitly demanding the value assignment is updated after any measurement. Further, we'd argue that a measurement update rule that updates value assignments is physically well justified: It is perfectly reasonable to demand that a measurement that extracts information from a system disturbs the state of the system, regardless of whether the measurement was performed with other commuting measurements. By inspecting the measurement update rules we can also see that we can only ever learn $n$-bits of information about the value assignment, due to the randomization step, satisfying the Holevo bound~\cite{Holevo1973}, similar to previously constructed contextual models of the Mermin-Peres square by Larsson and Kleinmann \emph{et al.}~\cite{Larsson2012,Kleinmann2011}.

To demonstrate how the model correctly reproduces the statistics of the Mermin-Peres square let us consider a test case of the $\ket{00}$ state, to be measured in one of the contexts. The model gives a deterministic outcome assignment to every Pauli observable, i.e. $v_\lambda(M) = (-1)^{\gamma_\lambda(M)}$. However, as we must update the ontic state after every measurement we can only effectively learn the values assigned to one context. Suppose the ontic state of the system is $\lambda = (Z\mathbb{I}, \vec{0}) \in \text{supp}(\mu_{00})$. We therefore have a value assignment to the observables in the Mermin-Peres square which can be expressed as;
\begin{align*}
\begin{array}{ccc}
 v_\lambda(M)=k & & v(M|\, \forall \lambda \in \text{supp}(\mu_{00})) \\
\begin{array}{ccc}
+ & + & + \\
+ & + & + \\
+ & + & +
\end{array}
&
\longleftrightarrow
&
\begin{array}{ccc}
\color{red} \{\pm\} & \color{red} \{\pm\}  & \color{red} \{\pm\}  \\
+ & + & + \\
\color{red} \{\pm\}  & \color{red} \{\pm\}  & \color{red} \{\pm\} 
\end{array}
\end{array}
\end{align*}

Where the second Mermin-Peres square has been used to demonstrate that cells not in the stabilizer group have uniformly random values. In the left-hand Mermin-Peres square the last column\rq{}s constraint is not satisfied.  So let us measure this problem context, starting with $YY$. To output an outcome of a measurement we simply read out the value of the phase function for that observable. So in this case $k_{YY} = +$. To perform the update rule we note that $[Z\mathbb{I},YY]\neq 0$ so we update $Z\mathbb{I} \rightarrow YY$ and follow the non-commuting cases in the update rules. The post-measurement update to the phase function is given by line 2 of equation \ref{eqn:2qUpdate}, so $\gamma^\prime(S) = \gamma(S)$ and  $\gamma^\prime(SM) = \gamma(S) + k +\beta(S,M), \, \forall S \in \ProjP{2}$ such that $[S,YY]=[S,Z\mathbb{I}] =0$. So we need to find the set of $2$ qubit Pauli operators that commute with both $Z\mathbb{I}$ and $YY$. Namely $\set{\mathbb{I}, ZZ, \mathbb{I}Y, ZX}$ which forms a non-abelian subgroup of the Pauli group\footnote{Note that the Mermin-Peres square gives graphical way to find such operators.}. Following the phase function update rules we need to update phase function to: 
\begin{align*}
\gamma^\prime(YY) &= \gamma(YY)= k= 0 ,\\
\gamma^\prime(ZZ)  &= \gamma(ZZ) = 0,\\
\gamma^\prime(XX) &= \gamma(ZZ) + k + \beta(ZZ,YY) = 1,\\
\gamma^\prime(ZX) &= \gamma(ZX) = 0,\\
\gamma^\prime(XZ) &= \gamma(ZX) + k + \beta(ZX,YY) = 0.
\end{align*}
Finally we randomize any elements that commute with $Z \mathbb{I}$ and anti-commute with $YY$, i.e. in this Mermin-Peres sqaure $\mathbb{Z} \mathbb{I}$, $\mathbb{I} Z$, and $\mathbb{I}X$. This means, after converting the phase function, the value assignment for the Mermin-Peres square after a $YY$ measurement is now, chosing the randomization to give $+$;
\begin{align*}
\begin{array}{ccc}
 v(M|\lambda^\prime)=k & & v(M|\, \forall \lambda^\prime \in \text{supp}(\mu_{B_{11}})) \\
\begin{array}{ccc}
+ & \color{ForestGreen}+ & \color{blue}- \\
\color{ForestGreen}+ & \color{ForestGreen}+ & + \\
\color{blue}+ & \color{blue}+ & +
\end{array}
&
\longleftrightarrow
&
\begin{array}{ccc}
\color{red} \{\pm\} & \color{ForestGreen} \{\pm\}  & \color{blue}-  \\
\color{ForestGreen} \{\pm\} & \color{ForestGreen} \{\pm\} & + \\
\color{blue} \{\pm\}  & \color{blue} \{\pm\}  & + 
\end{array}
\end{array}
\end{align*}

Where we have colored the deterministically updated cells blue and the randomized cells green. Note the final row\rq{}s value assignments are now correlated but random, this makes sense when we consider that given $P = Z\mathbb{I}$ the pre-measurement stabilizer group could have been $\set{\mathbb{I},Z\mathbb{I},\mathbb{I}Z,ZZ}$ or $\set{\mathbb{I},Z\mathbb{I},\mathbb{I}X,ZX}$, so the update rule must account for both possibilities. 

The update rule has enforced that the contexts that $YY$ is an element of satisfy the functional relationships, as the new ontic state is in the support of any state stabilized by $YY$. Further, as expected, the value assignments are consistent with the new stabilizer group $S(B_{11})=\set{\mathbb{I},-XX,ZZ,YY}$, where $\ket{B_{11}} = \frac{1}{\sqrt{2}} \left( \ket{00} - \ket{11}\right)$. If we were to go through the update rules for every other measurement in the context (the third column) we would see that the value assignment for the square only performs the randomization step.

\subsubsection[]{The PBR theorem}

One obstacle to constructing any ontological model of the stabilizer formalism is the Pusey-Barrett-Rudolph no-go theorem~\cite{PBR2012}. This no-go result states that any ontological model that treats preparations on subsystems independently of preparations on other subsystems, referred to as \emph{preparation independence}, must be $\psi$-ontic.

As the model presented here is $\psi$-epistemic the PBR theorem implies the model must represent preparations non-locally. This indeed is the case. For example consider the $\ket{00}$ preparation, the stabilizer group of $\ket{00}$ is given by $\mathcal{S}(\ket{00})=\set{\mathbb{I},Z\mathbb{I},\mathbb{I}Z,ZZ}$ and therefore the support of $\mu_{\ket{00}}$ only contains ontic states that satisfy $\gamma(Z\mathbb{I}) + \gamma(\mathbb{I}Z)=\gamma(ZZ)$, however the same cannot be said for other local variables such as $\gamma(XI)+\gamma(IZ) \neq \gamma(XZ)$. Hence the model evades the PBR theorem by being explicitly non-local, even for product preparations.

The anti-distinguishing measurements used in the PBR theorem also place additional constraints on the structure of possible $\psi$-epistemic models, as demonstrated by Karanjai \emph{et al.}.  Namely they imply that any set of states that can be anti-distinguished cannot have a global intersection\footnote{Karanjai \emph{et al.} also proved strict lower bound on the size of the ontic space of any ontological model of the stabilizer formalism, which the model presented here satisfies. This is covered in more detail in the appendix.}~\cite{Karanjai2018}. A quick sketch of the proof is as follows~\cite{Leifer2014}; Suppose we have a set of states $S=\set{\rho_i}_i$ that can be antidistinguished by some POVM, with elements $E_j$. So;
\begin{align*}
\text{Tr}(\rho_i E_j) \begin{cases} = 0, & \text{if} \, i=j, \\ \ge 0, &\text{otherwise.} \end{cases}
\end{align*}
I.e. if we measured outcome $j$ we know with certainty that $\rho_j$ was not prepared. This implies that in any ontological model that reproduces the statistics of an anti-distinguishing measurement $\xi_j(\lambda) = 0, \, \forall \lambda \in \text{supp}(\mu_{\rho_j})=\Delta_{\rho_j}$, as $\text{Tr}(\rho_j E_j) = \int_\Lambda d\lambda \, \xi_j(\lambda) \mu_{\rho_j}(\lambda) = \int_{\Delta_{\rho_j}} d\lambda \, \xi_j(\lambda) \mu_{\rho_j}(\lambda) = 0$. If we assume that $\mu_{\rho_i}$ share a joint support, $\exists \Delta \subset \Delta_{\rho_i}, \, \forall i$, then the same argument holds for all outcomes, i.e. $\xi_i (\lambda)=0, \, \forall \lambda \in \Delta, \,\forall i$. However by normalization we require $\sum_i \xi_i (\lambda)  = 1, \, \forall \lambda$, therefore we have a contradiction. 

While these anti-distinguishing POVMs cannot be measured by a single measurement in the stabilizer formalism, we can measure them via a sequence of adaptive measurements. To demonstrate how the model accounts for these sequences of adaptive measurements we use the anti-distinguishable set of states $\set{\ket{00},\ket{0+},\ket{+0},\ket{++}}$ from the  PBR paper~\cite{PBR2012}. The adaptive sequence of stabilizer measurements that anti-distinguishes this set of states is given in figure \ref{fig:PBR_Adapative_Measurements}.

\begin{figure}[t]
 \centering
    \includegraphics[width=0.5\textwidth]{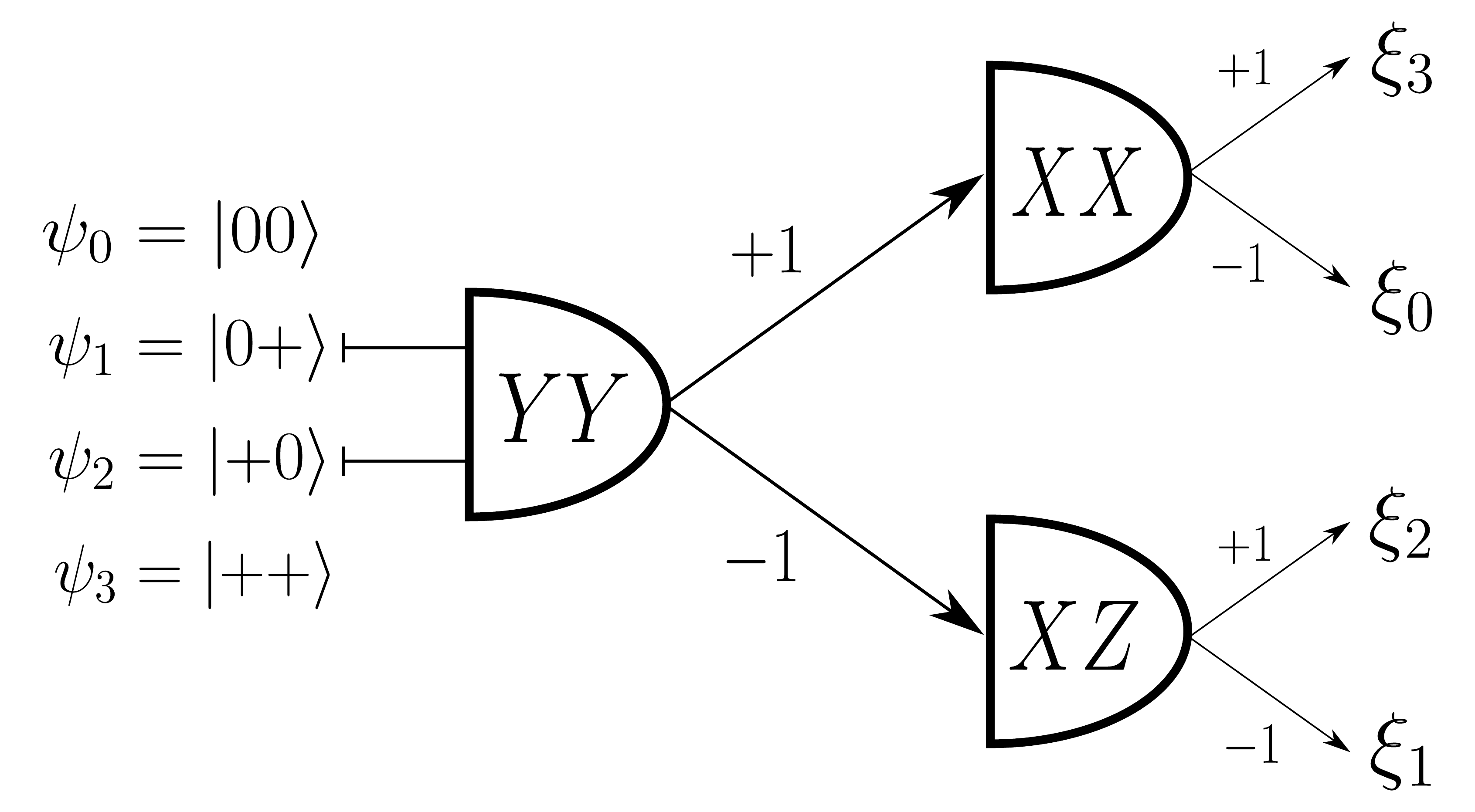}
\caption{The Pusey-Barrett-Rudolph anti-distinguishing adaptive measurement circuit. The input to the circuit is a $\psi_i$, which all give a random outcome for the measurement of $YY$. The second measurement is then conditioned on the outcome of $YY$. The possible sets of outcomes then anti-distinguishes the preparation. For example $\mathcal{O}(YY) = +1$ and $\mathcal{O}(XX) = -1$ implies $\psi_0$ wasn't prepared, hence this path is labelled $\xi_0$.}
\label{fig:PBR_Adapative_Measurements}
\end{figure}

The measurement of the first observable, with outcome $k$, updates the stabilizer groups of the input states as follows;
\begin{alignat}{4}\label{eqn:PBRStabUpdate}
\begin{array}{l} \ket{00} \\ \ket{0+} \\ \ket{+0} \\ \ket{++}  \end{array} & 
\begin{array}{l} \equiv \set{\mathbb{I}, \mathbb{I}Z, Z\mathbb{I}, ZZ} \\ \equiv \set{\mathbb{I}, \mathbb{I}X, Z\mathbb{I}, ZX} \\ \equiv \set{\mathbb{I}, \mathbb{I}Z, X\mathbb{I}, XZ} \\ \equiv \set{\mathbb{I}, \mathbb{I}X, X\mathbb{I}, XX}  \end{array} &
\scalebox{1.5}{$\xrightarrow[]{\mathcal{O}(YY) = k}$}&
\begin{array}{l}  \set{\mathbb{I}, ZZ, (-1)^k YY, (-1)^{k+1}XX}, \\  \set{\mathbb{I}, ZX, (-1)^k YY, (-1)^k XZ}, \\  \set{\mathbb{I}, XZ, (-1)^k YY, (-1)^k ZX}, \\  \set{\mathbb{I}, XX, (-1)^k YY, (-1)^{k+1} ZZ}.  \end{array}
\end{alignat}

By noting that two stabilizer states, $\rho_1$ and $\rho_2$, are orthogonal if and only if $\exists P \in \mathcal{S}(\rho_1)$ and $-P \in \mathcal{S}(\rho_2)$. We can see that the first and last post-measurement stabilizer states are orthogonal if $k=0$ and equivalent if $k=1$. Similarly, the second and third post-measurement groups are orthogonal if $k=1$  and equivalent otherwise. Therefore by adaptively measuring $XX$ or $XZ$, equally $ZZ$ or $ZX$, after the first measurement we can rule out a single preparation, anti-distinguishing the set of input states. 

The first thing to note about how the model reproduces the PBR statistics is that it satisfies the no joint overlap requirement. By investigating the stabilizer groups of the pre-measurement states above it is clear that the $4$ states do not share a mutual generator and therefore have no joint overlap. Mathematically, $\nexists P$ such that $(P, \cdot) \in \text{supp}(\mu_{\psi_i}), \, \forall i$ therefore $\cap_i \text{supp}(\mu_{\psi_i}) = \emptyset$. 

To demonstrate how the measurement update rule gives the correct anti-distinguishing measurement statistics we investigate the evolution of the representations of $\ket{00}$ and $\ket{++}$, a similar analysis can be given for $\ket{0+}$ and $\ket{+0}$. Firstly we let $\gamma(YY) = k$ for both states. Considering that the $\lambda =(P,\cdot)$ in the support of $\mu_{\ket{00}}$ and $\mu_{\ket{++}}$ are $P\in\set{\mathbb{I}Z,Z\mathbb{I},ZZ}$ and $P\in\set{\mathbb{I}X,X\mathbb{I},XX}$, respectively, therefore $\mu_{\ket{00}} \cap \mu_{\ket{++}} = \emptyset$. For completeness, we examine both $[P,YY] = 0$ and $[P,YY] \neq 0$.

Taking our input state to be $\psi = \ket{00}$ and $(P,\cdot) = (\mathbb{I}Z, \cdot)$ the update rules determine that we must update all $S \in \ProjP{2}\backslash \mathbb{I}$ such that $[S,YY]=[S,\mathbb{I}Z]=0$, as $[\mathbb{I}Z,YY]\neq 0$. This set is $\set{\mathbb{I}, ZZ, XZ, Y\mathbb{I}}$, therefore we perform the update;
\begin{align*}
 \gamma^\prime(ZZ) &= \gamma(ZZ) = 0, \\
\gamma^\prime (XX) &= \gamma(ZZ) + k + \beta(ZZ, YY) = k + 1, 
\end{align*}
as required by equation \ref{eqn:PBRStabUpdate}. Note we also update $\gamma^\prime (XZ) = \gamma(XZ)$ and $\gamma^\prime(ZX) = \gamma(XZ) + k + \beta(XZ,YY) = \gamma(XZ)+ k$, which as $\gamma(XZ)$ is random implies $\gamma^\prime(ZX)$ is random. This update is important because the ontic state we have considered is in the support of $\ket{+0}$, therefore if $\gamma(XZ)$ is fixed by consistency it gives the correct update for that state. Applying the above to $(\mathbb{I}X, \cdot) \in \text{supp}(\mu_{\ket{++}})$, where now the set of commuting operators is $\set{\mathbb{I},XX, ZX, Y\mathbb{I}}$, we have;
\begin{align*}
\gamma^\prime(XX) &= \gamma(XX) = 0, \\
\gamma^\prime(ZZ) &= \gamma(XX) + k + \beta(XX,YY) = k +1, 
\end{align*}
as required.

By symmetry we can infer the update rules for $\lambda_1 = (Z\mathbb{I}, \cdot)$ and $\lambda_2 = (X\mathbb{I}, \cdot)$ will follow a similar structure. So finally we must check the commuting cases $[ZZ,YY]=[XX,YY]=0$. By the update rules we know in this case we leave $\gamma(P)$ static and update $\gamma^\prime(PM)$, therefore given $\lambda_1 = (ZZ,\cdot)$ we perform the updates $\gamma^\prime(ZZ) = \gamma(ZZ)$ and $\gamma^\prime(XX) = \gamma(ZZ) + k + \beta(ZZ,YY) = k +1$ and for similarly for $\lambda_2 = (XX,\cdot)$ we have $\gamma^\prime(XX) = \gamma(XX)$ and $\gamma^\prime(ZZ) = \gamma(XX) + k + \beta(XX,YY) = k +1$, as required.

\subsubsection[]{Traditional Contextuality}

In the ontological model presented in this paper the phase function clearly defines a value assignment to all Pauli observables. However in comparison to the Mermin style proofs of contextuality ~\cite{Mermin1990,Peres1991,Mermin1993} we do not assume that the value assignments of all commuting observables can be simultaneously extracted. Further any measurement procedure necessarily disturbs the value assignments of commuting observables. Therefore as expected the model is Kochen-Specker contextual.

However we argue that the contextuality present in the stabilizer formalism is not a surprising form of contextuality. Taking the Mermin-Peres square as our example, there is no reason to believe that the outcome of $YY$ should be the same independent of the context it is measured in. Nor should its outcome be the same independent of where in the sequence of measurements it is performed. The only requirement imposed by the stabilizer formalism is that we satisfy the algebraic constraints, given by the stabilizer group operation, on observables in a stabilizer group. 

To illustrate this let us consider four measurement sequences $S_1 = XX\rightarrow ZZ\rightarrow YY$, $S_2 = XZ\rightarrow ZX \rightarrow YY$, $S_1^\prime = YY \rightarrow ZZ \rightarrow XX$, and $S_2^\prime = YY \rightarrow ZX \rightarrow XZ$. The first thing to note is that the first two measurements in each sequence define a rank-1 projective measurement, which fixes the outcome of the final measurement. For the first sequence we have the constraint $S_1: \, \mathcal{O}(XX)\mathcal{O}(ZZ) \oplus 1 = \mathcal{O}(YY)$, similarly $S_2: \,\mathcal{O}(XZ)\mathcal{O}(ZX) = \mathcal{O}(YY)$. The usual contradiction is derived by adding in contexts that jointly constrain the outcomes of $XX$, $ZZ$, $XZ$, and $ZX$. However this ignores what the algebraic constraints are imposing, i.e. the outcome of a $YY$ measurement is constrained by the outcomes of the previously measured observables. Not on the possible, counter-factual, measurement sequences. 

The final two sequences, $S_1^\prime$ and $S_2^\prime$, illustrate that such an update to the value assignment is physically justified. Consider measuring a context in the Mermin-Peres square with the maximally mixed state as the input. In this case the outcome of a $YY$ measurement should be random, by lemma \ref{lemma:StabStats}. If we measure sequence $S_1^\prime$ and $S_2^\prime$ the outcome $\mathcal{O}(YY)$ is clearly random, by definition of the phase function. Similarly if we measure sequences $S_1$ or $S_2$ the outcome $\mathcal{O}(YY)$ is still random, as the outcomes of the previous measurements in the sequence are random, but not necessarily equal to the outcome if $YY$ was measured first. Therefore there is no reason to expect the outcome of $YY$ to be the same regardless of whether it was measured first or last, as the only constraint we should obey is that it's outcome is random. Hence, the only case where we should expect a fixed outcome for $YY$ is if the input state was an eigenstate of $YY$, in which case the algebraic constraints limit our value assignments to other commuting observables, but crucially no others. 

\section[]{Generalizing the model to $n$-qubits}

To begin we first establish some notation, denote a proper abelian subgroup of the projective Pauli group with bold font, i.e. $\boldsymbol{G} \subset \ProjP{n}$. We define the set of all abelian subgroups with $k$ generators as $\boldsymbol{\mathcal{P}}_k := \set{\boldsymbol{G}| \, \boldsymbol{G} \subset \ProjP{n}, \, |\boldsymbol{G}| = 2^k}$. These subgroups will replace the role of $ \ProjP{2} \backslash \mathbb{I}$ in the two qubit model, and technically for the two qubit model the stabilizer ontology is $\boldsymbol{\mathcal{P}}_1$.

The ontology of the $n$-qubit model is defined to be $\Lambda = \boldsymbol{\mathcal{P}}_{n-1} \times \gamma_n$. The uniform distributions representing $n$-qubit pure stabilizer states $\psi \in \mathcal{S}(\mathcal{H}_{2^n})$ have support;
\begin{align}
\text{supp}(\mu_\psi) = \set{\lambda = (\boldsymbol{G},\gamma)| \, \boldsymbol{G} \subset \widetilde{\mathcal{S}}(\psi), \, \boldsymbol{G} \in \boldsymbol{\mathcal{P}}_{n-1}, \, \gamma \cong \psi }.
\end{align}
These states, while being defined on an exponentially growing ontology, are $\psi$-epistemic. For example the ontic state $\lambda = ({\left<Z_i\right>}_{i = 1,..,n-1},\gamma = \vec{0})$ is in the support of $\ket{0}^{\otimes n}$ and $\ket{0}^{\otimes n-1} \ket{+}$. To find which pure stabilizer states have support on a given ontic state note that given a $\boldsymbol{G} \in \boldsymbol{\mathcal{P}}_{n-1}$ there is a finite set of non-commuting Pauli operators that commute with $\boldsymbol{G}$ that can all be added to the group to construct valid stabilizer groups.

As before the Clifford operations in the model can be inferred from their action of the projective Pauli group. So a Clifford operation's, $C$, ontological representation, $\Gamma_C$, is given by;
\begin{align}
\Gamma_C: (\boldsymbol{G}, \gamma) \mapsto (C\boldsymbol{G}C^\dagger, \gamma^R + \gamma_c),
\end{align}
Where again we re-order $\gamma$ to match the Clifford permutation and the conjugation is take to be on every element of $\boldsymbol{G}$.

As in the two qubit model the response functions read-out the value of the phase function;

\begin{align}
\xi_{k, M} (\lambda) = \begin{cases}
1 & \text{if} \, \, \gamma(M) = k, \\
0 & \text{otherwise.}
\end{cases}.
\end{align}

We define the measurement update rules for a measurement of $M$ with outcome $k$ to be;

\begin{align}\label{eqn:nqUpdate}
\Gamma_{k|M}: (\boldsymbol{G},\gamma) \mapsto (\boldsymbol{G}^\prime, \gamma^\prime)
\begin{cases}
\boldsymbol{G}\mapsto \boldsymbol{G}^\prime = \left<\boldsymbol{G}_M,M\right> \, \text{w. eq. prob} & \begin{array}{l} \boldsymbol{G}_M \in \boldsymbol{\mathcal{P}}_{n-2} \, | \, [M,G_i] = 0, \\ \boldsymbol{G}_M \subset \boldsymbol{G}, \, \forall G_i \in \boldsymbol{G}_M .\end{array}\\
\text{if} \, [M,\boldsymbol{G}] \neq 0:\left.\begin{array}{l}\gamma^\prime(S) = \gamma(S) \\ \gamma^\prime(SM) = \gamma(S) + k + \beta(S,M) \end{array}\right\} & \forall S\in \ProjP{n}, \, | \,  [S,M]=[S,\boldsymbol{G}] = 0. \\
\text{if} \, [M,\boldsymbol{G}]=0: \, \,\gamma^\prime(G_iM) = \gamma(G_i) + k + \beta(G_i,M)  & \forall G_i \in \boldsymbol{G} .\\ 
\gamma^\prime(S) = \set{0,1} \, \text{w. eq. prob} & \begin{array}{l}\forall S\in \ProjP{n},  \, | \, [S,M] \neq 0, \, [S,G_i] =0, \\ \forall G_i \in \boldsymbol{G}.\end{array}\\
\gamma^\prime(S) = \gamma(S) & \text{Otherwise.}
\end{cases}
\end{align}

\begin{theorem}\label{theorem:Main}
The ontological model presented above is a $\psi$-epistemic model of the $n$-qubit stabilizer subtheory.
\end{theorem}

\begin{proof}

The model is $\psi$-epistemic by the example given above. It also reproduces the single shot statistics, by the argument given for the prepare-measure-discard model. Therefore we just to check that the measurement update rule maps the phase function to a new phase function that is consistent with the post-measurement state.

The first line corresponds to updating $\boldsymbol{G}$ to a new subgroup that is an proper subgroup of the post-measurement state. As there are many possible proper subgroups of $\boldsymbol{G}$ that commute with $M$ if $[M,G_i]=0, \, \forall G_i\in \boldsymbol{G}$ we sample one of them uniformly. If $M$ does not commute with all elements of $\boldsymbol{G}$ there is only one such $\boldsymbol{G}_M$. $\left<\boldsymbol{G}_M,M\right>$ is an proper subgroup of the post-measurement group and therefore $(\boldsymbol{G}^\prime,\cdot) \in \text{supp}\mu_{\rho^\prime}$

We will prove correctness by running through each of the possible measurement updates that need to be applied, which are categorized by the commutation relation between $\boldsymbol{G}$ and $M$. 

The first case we consider is $[M,\boldsymbol{G}]\neq 0$, i.e.  $M$ does not commute with $2^{n-2}$ elements of $\boldsymbol{G}$. We can write the post-measurement group as $\left<\boldsymbol{G}_ M,M,S \right>$, where $S$ is the unknown stabilizer element, such that $[S,M]=[S,\boldsymbol{G}]=0, \, S\notin \boldsymbol{G}$ and $\boldsymbol{G}_ M$ is the subgroup of $\boldsymbol{G}$ containing all elements that commute with $M$. There are $3$ non-trivial equivalence classes of $S$, which mutually anti-commute. Any $S \in \ProjP{n}$ that commutes with both $M$ and $\boldsymbol{G}$ could have been elements of the pre-measurement and post-measurement group, implying their phase function should remain fixed through the measurement. We therefore update all possible new elements of the post-measurement group, which are $SM \notin \mathcal{S}(\rho)$ as $\exists G \in \boldsymbol{G}$ such that $[SM, G] \neq 0$. Given two possible extensions generated by $S_1$ and $S_2$ such that $[S_1,S_2]\neq 0$ we have $[S_1 M g, S_2 M g^\prime] \neq 0, \forall g, g^\prime \in \boldsymbol{G}_M$. This implies the post measurement equivalence classes do not intersect, i.e. $S_1 M g \notin \left<G_M, M, S_2\right>, \, \forall g \in G_M$. Therefore $SM$ is only an element of one possible group extension implying that $\gamma^\prime(SM)$ has the correct phase if $SM \in \mathcal{S}(\rho^\prime)$. Note that any $G \in \boldsymbol{G}$ that commutes with $M$ satisfies the conditional of lines 2 and 3. Finally any element that could have been part of the pre-measurement group, but not the post-measurement group is randomized, line 4.

The second case we need to consider is when $M$ commutes with all elements of $\boldsymbol{G}$. Again this can be broken down to two cases $M \notin \boldsymbol{G}$ and $M \in \boldsymbol{G}$. If $M$ commutes with all of $\boldsymbol{G}$ but is not an element of the group, then we know the post-measurement stabilizer group can be expressed as $\left<\boldsymbol{G},(-1)^k M \right>$. Therefore the new (or potentially old) elements of the group are $G_i M$ with phase $\gamma(G_i M) = \gamma(G_i) + k +\beta(G_i,M) = \gamma^\prime(G_i M)$. Note if these Pauli operators were elements of the stabilizer group prior to measurement their phase remains unchanged. If $M \in \boldsymbol{G}$ then the phase function should remain unchanged through the measurement, up to randomization on stabilizer measurements not in the stabilizer group. This is expressed in lines 3 through 5 of the update rule. Note line 5 will be satisfied by operators that commute with both $M$ and $\boldsymbol{G}$, but are not in $\boldsymbol{G}$, these are possible elements of the stabilizer group and therefore their phase function should remain fixed.

Therefore the update rule given above maps an ontic state $\lambda = \left( \boldsymbol{G}, \gamma \right) \in \text{supp}(\mu_\psi)$ to an ontic state $\lambda^\prime = \left(\boldsymbol{G}^\prime,\gamma^\prime \right) \in \text{supp}(\mu_{\psi^\prime})$, where $\rho_{\psi^\prime} \propto \frac{\mathbb{I} + (-1)^k M}{2} \rho_\psi \frac{\mathbb{I} + (-1)^k M}{2}$. Further the randomization procedure ensures a random value of the phase function for observables not in the post-measurement stabilizer group. Therefore the model correctly reproduces the stabilizer statistics.

\end{proof}

It is interesting to note that the additional stabilizer ontology we introduced has not been required prior to defining the update rules. However, we can actual use this ontology to streamline the model to define an always-$\psi$-epistemic model. We can achieve this by noticing that when $[\boldsymbol{G},M] \neq 0$ the outcome of the measurement should be random, and therefore we can reconstruct the response functions accordingly. An overview of this procedure is covered in the appendix.

The generalization of the model to $n$-qubits has required us to extend the stabilizer ontology to proper subgroups of maximal abelian subgroups of the Stabilizer operators. It is not clear whether this additional ontology is required or can be made more compact. However by reducing the size of the groups in the ontology leads to ambiguities in how the phase function should be updated during measurement. It is therefore a highly interesting open problem whether $n-1$ generators of a stabilizer state are truly required to correctly update an epistemic state in an arbitrary $\psi$-epistemic model of the $n$-qubit stabilizer formalism. If this is indeed the case, it would suggest \emph{almost} full knowledge of a stabilizer state\footnote{As $n$ generators uniquely specify the state.} is required to be encoded in the ontology to correctly reproduce the stabilizer subtheory\rq{}s statistics.

\subsection[]{The $n$-qubit model's relation to the single-qubit 8-state model}

\begin{corollary} The $n=1$ case of the model is equivalent to the $8$-state model.
\end{corollary}

\begin{proof} To see this let us consider the ontology for a single qubit in the model is $\Lambda = \boldsymbol{\mathcal{P}}_0 \times \gamma_{n=1}$. As $ \boldsymbol{\mathcal{P}}_0 = \{\mathbb{I}\}$ and all stabilizer groups contain $\mathbb{I}$, the stabilizer group part of the ontology is trivial and we can remove it. The non-trivial ontology is therefore $\Lambda = \gamma_1$, which is identical to the $8$ possible value assignments over $X$, $Y$, and $Z$ in the 8-state model.

States are defined as having uniform support over all consistent phase functions. Writing the phase function as $\gamma = (0,x,y,z)$ we see that a stabilizer state $\rho = \frac{\mathbb{I} + (-1)^{s_P} P}{2}$ has support on all phase functions such that $\gamma(P) = s_P$. I.e. a stabilizer state has support on $4$ ontic states as in the 8-state model, and the supports in coincide. Therefore stabilizer states in the model presented here and the 8-state model have the same representation.

Clifford transformations are represented as permutations of the ontic space in both our model and the 8-state model. In the 8-state model the permutation maps representing a Clifford operations are derived from their stabilizer relations, i.e. $\Gamma_{8-\text{state}}(H): (x,y,z) \mapsto (z, y\oplus 1,x)$. This also holds in the model presented here, i.e. $\Gamma_{\text{model}} (H): (x,y,z) \mapsto (z,y,x) \oplus (0,1,0)$, where we have decomposed the map into $\gamma^R$ and $\gamma_c$ for clarity.

The measurement response functions are defined identically, with the response function essentially reading out the value of the phase function/value assignment. All that is left is to demonstrate is that the update rules are the same. In the single-qubit stabilizer formalism all measurements are rank-$1$ projectors. Therefore the measurement update map in the 8-state model can be simply stated as repreparing the eigenstate corresponding the eigenvalue measured. For example an $X$ measurement with a $+1$ outcome the measurement update rule prepares the distrubution $\mu_{\ket{+}}$. 

Returning to the $n=1$ case of the updates rules for the model, equation \ref{eqn:nqUpdate}. As all single qubit measurements commute with $\boldsymbol{\mathcal{P}}_0$ ontology we only consider the commuting case of the update rules. In this case we have $\gamma^\prime(G_i M) = \gamma(G_i) + k + \beta(G_i, M), \, \forall G_i \in \boldsymbol{G}$, which for a single qubit reduces to setting $\gamma^\prime(M) = k$. Finally we randomize all phase function elements that don't commute with $M$, which in the case of a single qubit Pauli observable is all other single qubit Pauli observables. This update effectively reprepares the eigenstate of the measured observable with eigenvalue $k$. Therefore the update rules are the same. This implies the $8$-state model and the $n=1$ case of the ontological model of the qubit stabilizer formalism presented here are the equivalent.
\end{proof}

\section{Discussion}

In this paper we have constructed a contextual $\psi$-epistemic model of the $n$-qubit stabilizer formalism. The model is constructed by considering that given a stabilizer state it is possible to reproduce the one-shot statistics of stabilizer measurements by uniform sampling over value assignments, which we term the phase function, that are consistent with a stabilizer state. Measurement update rules then ensure that this value assignment is mapped to a value assignment consistent with the post-measurement state. Therefore reproducing the statistics of the $n$-qubit stabilizer formalism. 

The core component of the model's presentation is the value assignment. We have constructed the model this way to investigate the structure of contextuality in the $n$-qubit stabilizer formalism. However, to correctly update the value assignment through a measurement we are need to encode $n-1$ generators of the stabilizer state in the ontology. Whether this additional ontology is required or can be reduced is an interesting open question. 

By moving away from an outcome deterministic model it is possible to use the presented model to construct an \emph{always}-$\psi$-epistemic model, covered in the appendix. Always-$\psi$-epistemic models have a strong resemblance to probabilistic classical models as a quantum state represents a true lack of knowledge of the physical state of a system. Therefore they present an interesting possibility of a categorizing the transition from classical to quantum theory. Further they share many properties with weak-simulations and could providing an intriguing formalism for investigate the resources required to simulate quantum systems, universal or otherwise.

\begin{acknowledgments}
We would like to thank Hammam Qassim, Juan Bermejo-Vega, and Joshua Ruebeck for useful discussions. This research was supported by the Government of Canada through CFREF Transformative Quantum Technologies program, NSERC Discovery program, and Industry Canada.
\end{acknowledgments}

\bibliography{UnderlyingOntology}

\appendix

\section{Proof of correctness of the 2-qubit ontological model}\label{sec:2qCorrectness}

Here we give the proof of lemma \ref{lemma:2qModel};

\begin{proof}

From the argument given for the prepare-measure-discard model we know that we will correctly reproduce the stabilizer outcome statistics if the phase function satisfies $\gamma(M)=p_b, \, \forall b \in \mathcal{B}(\mathcal{S}(\rho))$ and is random otherwise. The distributions representing preparations clearly satisfy this, therefore we just need to show the update rule maps $\gamma$ to a $\gamma^\prime$ such that $\gamma^\prime(M)=p_{b^\prime}, \, \forall b^\prime \in \mathcal{B}(\mathcal{S}(\rho^\prime))$ and random otherwise. Additionally by the stabilizer update rule we know that $(-1)^k M \in \mathcal{S}(\rho^\prime)$ so the update on the stabilizer operator portion of the ontology is in the support of the post-measurement state. 

To prove the phase function update map correct updates to a phase function consistent with the post-measurement state we will exhaustively prove each possibility, with each case defined via the commutation relation between $M$ and $P$.

For the first case, we let $[P,M] \neq 0$, i.e. $M$ was certainly not an element of the pre-measurement stabilizer group. We can always write this group as;
\begin{align*}\set{\mathbb{I},(-1)^{\gamma(P)}P, (-1)^{\gamma(S)}S , (-1)^{\gamma(S)+\gamma(P)}PS},\end{align*} where $[M,PS]=[M,P] \neq 0$ and $[S,M] =0$. Therefore the post measurement group is given by;
\begin{align*}\set{\mathbb{I}, (-1)^{\gamma(S)}S, (-1)^{k}M, (-1)^{\gamma(S) + k + \beta(S,M)}P_{S+M}}.\end{align*}
I.e. we need to keep $\gamma^\prime(S) = \gamma(S)$ static through the update and update the phase function on $P_{S+M}$ to $\gamma^\prime(SM) = \gamma(S) + k + \beta(S,M)$. This is captured by lines 2 and 3 of the update rule. However in the update rule we search over all such $S$, we do this because given a $P$ there are 3 possible non-trivial stabilizer groups $P$ is an element of, modulo phases. I.e. there are $3$ non-commuting Pauli operators, $\{S_1,S_2,S_3\}$ we could combine with $P$ to construct all possible maximal stabilizer groups containing $P$.  These groups only share the elements $\mathbb{I}$ and $P$, this can easily be verified by checking $PS_i$'s commutation relations with $S_{j\neq i}$ and $PS_{j\neq i}$, therefore by updating all possible $S$ we cover all possible pre-measurement stabilizer groups.  Further by the consistency of the phase function with the pre-measurement state we know if $S\in\mathcal{S}(\rho)$ then $\gamma(S)=p_S$, and if not $\gamma(S)$ is a random value. Therefore the phase function remains random on elements not in the group. Finally for any $S$ such that $[P,S]=0$ and $[M,S]\neq 0$ we randomize $\gamma(S)$ as $S$ could have been an element of the pre-measurement group, but is definitely not an element after. Therefore $\gamma^\prime$ is consistent with all possible post-measurement stabilizer groups.

The second case is that $[M,P] = 0$ and $M\neq P$. Regardless of what the pre-measurement group was we can actually directly construct the post-measurement group;
\begin{align*}
\set{\mathbb{I},(-1)^{\gamma(P)}P,(-1)^k M, (-1)^{\gamma(P) + k + \beta(P,M)} P_{P+M}}.
\end{align*} 
Therefore we only need to update $\gamma^\prime(PM) = \gamma(P) + k + \beta(P,M)$ and randomize all phase function elements that could have been elements of the pre-measurement group, but cannot be part of the post-measurement group, lines 4 and 5 of the update rule. Finally we have the case that $P = M$. This case is trivial and no update rule should be applied as pre-and-post measurement stabilizer groups are the same. However this case is rolled into the previous case in the update rule above, this can be seen by noting $\gamma^\prime(PM) = \gamma(MM) = \gamma(\mathbb{I}) =0$ and the randomizing step does not effect any potential stabilizer elements.

Therefore the update rule $\Gamma_{k|M}$ maps a phase function $\gamma$ to a set of phase functions $\set{\gamma^\prime_i}_i$ such that $\gamma^\prime_i(P_{b^\prime}) = p_{b^\prime}, \, \forall i, \, \forall b \in \mathcal{B}(\mathcal{S}(\rho^\prime))$ and $\gamma^\prime_i(M)$ is uniformly random over $i$ if $\pm P_{b^\prime} \notin \mathcal{S}(\rho^\prime)$. And therefore reproduces the stabilizer statistics for sequential measurements.

\end{proof}

\section{Streamlining the model}

In the main body of this paper we have presented the ontological model of the stabilizer formalism in such a way that the phase function encodes the stabilizer statistics. This has been done for ease of presentation. However, we can reduce the size of the model by utilizing the stabilizer part of the ontology.

Recall the ontology of the model is defined to be $\Lambda_n = \boldsymbol{\mathcal{P}}_{n-1} \times \gamma_n$. This definition of the ontology of the model grows super-exponentially;
\begin{align*}
|\Lambda_n| = |\boldsymbol{\mathcal{P}}_{n-1}| \times |\gamma_n|= \left[2^{(n-1)(n-2)/2}\prod_{k=0}^{n-2} (4^n -1) \right] \times 2^{4^n},
\end{align*}
where we have used a similar argument as ~\cite{Aaronson2004} for the first term. Clearly $2^{4^n}$ dominates for all but the smallest $n$. Therefore the ontology grows super-exponentially with number of qubits, making it a poor candidate for a physically motivated ontological model. 

The super-exponential scaling derives from us allowing all possible value assignments the model. However if consider the stabilizer statistics, lemma \ref{lemma:StabStats}, we see that $\exists S \in \mathcal{S}(\rho)$ such that $[S,M] \neq 0$ then the outcome of a measurement of $M$ should be random. 

Therefore if we have an ontic state $\lambda = \left( \boldsymbol{G}, \gamma \right)$ then we know the outcome of any measurement $M$ will be random if $[\boldsymbol{G},M]\neq 0$. Further this randomness can be encoded in the response functions, rather than the phase function;
\begin{align*}
\xi_{k, M} (\lambda) = \begin{cases}
1 & \text{if} \, \, \gamma(M) = k, \, [\boldsymbol{G},M]=0 \\
\frac{1}{2} & \text{if} \, \, [\boldsymbol{G},M]\neq 0\\
0 & \text{otherwise.}
\end{cases}
\end{align*}
Therefore we do not need to store the value of the phase function if a Pauli operator does not commute with $\boldsymbol{G}$.

Further as we know the phase function must be consistent with a stabilizer state we can infer that we only need to store the phase on a set of generators of $\boldsymbol{G}$, due to the group structure of the stabilizer groups. I.e. it must be the case that $\gamma(g_1 + g_2) = \gamma(g_1) + \gamma(g_2) + \beta(g_1,g_2), \, \forall g_1, g_2 \in \boldsymbol{G}$. 

Finally as $\boldsymbol{G} \in \boldsymbol{\mathcal{P}}_{n-1}$ there are only three possible equivalence classes of generators we can add to $\boldsymbol{G}$ to construct a pure stabilizer state's stabilizer group. Denote these possible extensions as $\mathcal{S}(\psi_1) = \left<\boldsymbol{G},S_1 \right>$, $\mathcal{S}(\psi_2) = \left<\boldsymbol{G},S_2 \right>$, and $\mathcal{S}(\psi_3) = \left<\boldsymbol{G},S_3 \right>$, where $\mathcal{S}(\psi_1) \neq \mathcal{S}(\psi_2) \neq \mathcal{S}(\psi_3)$. As the phase function must be consistent with one of these three group extensions, we only need to store the phase of the relevant generator, $\gamma(S_i)$. However if we wish not to be able to infer the quantum state from the ontic state we should treat each extension equivalently. So we demand the phase function also satisfies $\gamma(S_i + g) = \gamma(S_i) + \gamma(g) +\beta(S_i, g), \forall S_i$, which is satisfiable via the non-intersecting property of $\set{\mathcal{S}(\psi_i)\backslash \boldsymbol{G}}_i$.

Therefore can construct an \emph{symmetric-always}-$\psi$-epistemic model by encoding the ontic state in a tableau similar to the Gottesman-Aaronson strong-simulation ~\cite{Aaronson2004};

\begin{align}
\lambda \equiv 
\begin{array}{c|cccc | cccc | c}
\multirow{4}{*}{$\set{g_i}_i$}&a_{x_1,g_1}&a_{x_2,g_1}& \cdots & a_{x_n,g_1} & a_{z_1,g_1}&a_{z_2,g_1}& \cdots & a_{z_n,g_1} & \gamma(g_1) \\
&a_{x_1,g_2}&a_{x_2,g_2}& \cdots & a_{x_n,g_2} & a_{z_1,g_2}&a_{z_2,g_2}& \cdots & a_{z_n,g_2} & \gamma(g_2) \\
&\vdots&\vdots& \ddots & \vdots & \vdots&\vdots& \ddots & \vdots & \vdots \\
&a_{x_1,g_{n-1}}&a_{x_2,g_{n-1}}& \cdots & a_{x_n,g_{n-1}} & a_{z_1,g_{n-1}}&a_{z_2,g_{n-1}}& \cdots & a_{z_n,g_{n-1}} & \gamma(g_{n-1}) \\
\hline 
\multirow{3}{*}{$\set{S_i}_i$}&a_{x_1,S_1}&a_{x_2,S_1}& \cdots & a_{x_n,S_1} & a_{z_1,S_1}&a_{z_2,S_1}& \cdots & a_{z_n,S_1} & \gamma(S_1) \\
&a_{x_1,S_2}&a_{x_2,S_2}& \cdots & a_{x_n,S_2} & a_{z_1,S_2}&a_{z_2,S_2}& \cdots & a_{z_n,S_2} & \gamma(S_2) \\
&a_{x_1,S_3}&a_{x_2,S_3}& \cdots & a_{x_n,S_3} & a_{z_1,S_3}&a_{z_2,S_3}& \cdots & a_{z_n,S_3} & \gamma(S_3) \\
\end{array}
\end{align}
Where the binary-sympletic representation has been used to encode Pauli operators.
From this we can see that an ontic state can be stored in $(2n+1)(n+2)$ bits, which satisfies the bound proven by Karanjai \emph{et al.} ~\cite{Karanjai2018}. The ontology is symmetric as each ontic state is in the support of exactly $3$ quantum states, given by the stabilizer groups $\mathcal{S}(\psi_1) , \, \mathcal{S}(\psi_2), \, \text{and} \, \mathcal{S}(\psi_3)$.

This construction drastically reduces the size of the ontology of the model. It also provides a route to investigating whether such a model can be used as the basis for a weak-simulation scheme, which will be covered in a follow up paper.

\section[]{Generalized Contextuality}

Generalized contextuality is an extension of traditional contextuality to preparations, transformations, and probabilistic ontological models. This generalization is defined via the concept of operational equivalences. We say that two (experimental) physical operations are \emph{operationally equivalent}, denoted $\cong$, if and only if they give the same experimental statistic, regardless of the choice of experimental procedure:
\begin{tabular}{p{2.65cm}p{12cm}}
\emph{Preparations:} & Two preparation procedures $P$ and $P^\prime$ are operationally equivalent, $(P\cong P^\prime)$, \emph{iff} $\text{Pr}(k|P,T,M) = \text{Pr}(k|P^\prime,T,M), \, \forall T,M$. \\
\emph{Transformations:} & Two transformation procedures $T$ and $T^\prime$ are operationally equivalent, $(T\cong T^\prime)$, \emph{iff} $\text{Pr}(k|P,T,M) = \text{Pr}(k|P,T^\prime,M), \, \forall P,M$.  \\
\emph{Measurements:} & The outcome $k$ of two measurement procedures $k\in M$ and $k \in M^\prime$ is operationally equivalent, $([k,M]\cong [k,M^\prime])$, \emph{iff} $\text{Pr}(k|P,T,M) = \text{Pr}(k|P,T,M^\prime), \, \forall P,T$.
\end{tabular}
\\

The assumption of generalized non-contextuality then states that any two operationally equivalent procedures in our physical theory should be represented by the same object in an ontological model. Therefore we say an ontological model is \emph{preparation non-contextual} (PNC) if;
\begin{align}
\mu_P = \mu_{P^\prime} \, \Leftrightarrow \, P \cong P^\prime.
\end{align}
Similarly it is \emph{transformation non-contextual} (TNC) if;
\begin{align}
\Gamma_T = \Gamma_{T^\prime} \, \Leftrightarrow \, T \cong T^\prime.
\end{align}
And \emph{measurement non-contextual} (MNC) if;
\begin{align}
\xi_{k,M} =\xi_{k,M^\prime} \Leftrightarrow \, [k,M] \cong [k,M^\prime].
\end{align}
As previously mentioned the $n$-qubit stabilizer formalism exhibits all forms of contextuality. Therefore it is not surprising that the model presented in this paper does not satisfy \emph{any} of these requirements. Hence it is a preparation-transformation-measurement contextual ontological model.

\subsection{The model and generalized contextuality}

\subsubsection[]{Preparation Contextuality}

The model is preparation contextual for $n>1$. To see this we use what we call \emph{Pauli} eigenbases. A Pauli eigenbasis $\set{\psi_i}_i$ is a set of stabilizer states such that we can write the stabilizer groups of all eigenstates as $\mathcal{S}(\psi_i) =  \left<\pm G_1, \pm G_2, ..., \pm G_n \right>$, where $\{G_i \in \ProjP{n}\}$ are the same for all eigenstates. We call these Pauli eigenbases as they are exactly the joint eigenstates of sets of commuting Pauli operators. However it should be noted that these are not the only orthonormal bases we can construct in the stabilizer formalism. The PBR POVM elements give a clean example of an orthonormal basis that does not have the structure of a Pauli eigenbasis\footnote{To construct these bases you leverage the fact to stabilizer states $\rho_1$ and $\rho_2$ are orthogonal \emph{iff} $P \in \mathcal{S}(\rho_1)$ and $-P \in \mathcal{S}(\rho_2)$.}.

So without loss of generality let us consider the two Pauli eigenbases;
\begin{align*}
\set{\psi_{Z,j}} &\cong \set{\mathcal{S}(\psi_{Z,j})=\left<(-1)^{f(j,i)} Z_i \right>_{i=1,..,n}}, \\
\set{\psi_{X,k}} &\cong \set{\mathcal{S}(\psi_{X,k})=\left<(-1)^{f(k,i)} X_i \right>_{i=1,..,n}},
\end{align*}
where $Z_i = \mathbb{I}_{j\neq i}\otimes Z_i$ and similarly for $X_i$, and $(-1)^{f(i,j)}$ is a function that maps  index $j$ to all possible $\pm 1$ phases on each generator. These are the $Z$ (computational) and $X$ eigenbases respectively and clearly share no non-identity stabilizer element, i.e. $\nexists S,j,k$ such that $S \in \mathcal{S}(\psi_{Z,j})\backslash \mathbb{I}$ and $S \in \mathcal{S}(\psi_{X,k})\backslash \mathbb{I}$. Therefore by noting the definition of the support of a stabilizer state in the model we can infer that $\mu_{\psi_{Z,j}} \cap \mu_{\psi_{X,k}} = \emptyset, \,\forall j,k$.

To demonstrate that the model is preparation contextual consider that both eigenbases can be used to prepare the maximally mixed state; 
\begin{align*}\frac{1}{2^n}\mathbb{I} = \sum_j \frac{1}{2^n} \ketbra{\psi_{Z,j}}{\psi_{Z,j}} = \sum_k \frac{1}{2^n} \ketbra{\psi_{X,k}}{\psi_{X,k}} =\frac{1}{2^n}\mathbb{I}.\end{align*}
Therefore the representation of the maximally mixed state depends on which eigenbasis we perform the preparation in;
\begin{align*} \mu_{\mathbb{I}/2^n}^{(Z)} = \sum_j \frac{1}{2^n} \mu_{\psi_{Z,j}} \neq \sum_k \frac{1}{2^n} \mu_{\psi_{X,k}} = \mu_{\mathbb{I}/2^n}^{(X)},\end{align*}
by the disjointness of the pure state preparations. Therefore the model is preparation contextual. Strangely, this implies $\text{supp}(\mu_{\mathbb{I}/2^n}^{(Z)}) \cap \text{supp}(\mu_{\mathbb{I}/2^n}^{(X)}) = \emptyset$. However, we can define a \emph{cannonical} representation of the maximally mixed state via a uniform distribution over all possible representations.

The above analysis does not apply in the $n=1$ case as the stabilizer part of the ontology is trivial, i.e. for a single qubit $\Lambda_{n=1} = \set{(\mathbb{I}, \gamma)| \gamma \in \gamma_{n=1}}$ and $\mathbb{I} \in \mathcal{S}(\rho), \, \forall \rho \in \mathcal{S}(\mathcal{H}_2)$. Therefore we can effective discard it, making the $n=1$ case preparation non-contextual as all bases span the space of phase functions, note this spanning nature of a basis also holds for $n>1$.

\subsubsection[]{Transformation Contextuality}

The model inherits the transformation contextuality of the single qubit stabilizer subtheory ~\cite{Lillystone2018}. This implies, via embedding, the model is transformation contextual for all $n$. 

Recently there has been proposals to strengthen the definition of transformtion contextuality to be more akin to traditional contextuality ~\cite{Mansfield2018}. This definition would class an operational theory as transformation contextual if the representation of a unitary transformation was dependent on the set of unitaries it was performed with. Under this definition the model is transformation non-contextual.

\subsubsection[]{Measurement Contextuality}

The model up to this point has been constructed such that only rank $2^{n-1}$ measurements are permissible, i.e. binary outcome measurements. This restriction has been imposed on the model as it directly corresponds to the measurements allowed within the stabilizer formalism. However any model of binary outcome measurements are trivially measurement non-contextual, as there is only one context a projector can be measured in.

Considering, a sequence of rank-$2^{n-1}$ PVMs can be used to construct rank-1 PVMs. For example we can construct a computational basis measurement for two-qubits by measuring $Z\mathbb{I}$, with outcome $k_1$, followed by $\mathbb{I}Z$, with outcome $k_2$, which will project any state onto $\ket{k_1 k_2}$ with the correct probability. Alternatively we could have measured $\mathbb{Z}I$ followed by $ZZ$, with outcome $k_3$, which will project onto the state $\ket{k_1 (k_3 \oplus k_1)}$. This gives two contexts in which the computational basis can be measured.

By using the operational elements of the model we can construct \emph{effective} response functions that correspond to measuring any rank $2^k, \, k \in \set{1,...,n-1}$ observable with stabilizer eigenspaces. Here we will present this construction for two sequential measurements. However, it can easily be extended to longer measurement sequences. We consider a sequence of two measurements $M_1 \rightarrow M_2$ such that $[M_1,M_2] =0$, this implies that both are diagonal in some basis. Letting the outcome of each measurement be $k_1$ and $k_2$, respectively, we can construct effective response function for the joint measurement via;
\begin{align}
\xi_{(M_1 \rightarrow M_2)} (k_1,k_2| \lambda) &= \text{Pr}(k_1,k_2|M_1,M_2,\lambda)=\text{Pr}(k_2| k_1, M_1, M_2, \lambda)\text{Pr}(k_1|M_1,\lambda), \nonumber\\
&= \int_\Lambda d\lambda^\prime \text{Pr}(k_2,\lambda^\prime| k_1, M_1, M_2, \lambda)\text{Pr}(k_1|M_1,\lambda), \nonumber\\
&= \int_\Lambda d\lambda^\prime \text{Pr}(k_2|M_2,\lambda^\prime)\text{Pr}(\lambda^\prime|k_1,M_1,\lambda)\text{Pr}(k_1|M_1,\lambda), \nonumber\\
&= \int_\Lambda d\lambda^\prime \xi_{M_2} (k_2 | \lambda^\prime) \Gamma_{k_1 | M_1} (\lambda^\prime, \lambda) \xi_{M_1}(k_1|\lambda), \label{eqn:EffRepFn}
\end{align}
where $\xi_{M_i}(k_i | \lambda)$ are the response functions for a binary outcome measurements and $\Gamma_{k|M}$ is the post-measurement transition map. Therefore the effective response functions are dependent on the sequence of measurements and choice of rank-$2^{n-1}$ measurements used. Implying the model is measurement contextual.

Alternatively, to see how the model must be measurement contextual consider that traditional non-contextuality is the conjunction of measurement non-contextuality and outcome determinism. Therefore as the model is outcome deterministic and traditionally contextual, it must be measurement contextual\footnote{A final possible way to derive measurement contextuality is to note a projector can be associated to an outcome in many different measurement procedures. For example, $\Pi_{00}=\ketbra{00}{00}$ is an a possible outcome of measuring each qubit in the computational basis, e.g. $M_1=Z\mathbb{I}$ and $M_2=\mathbb I Z$, and it is an outcome of a PBR-style measurement, e.g. $M_1=ZZ$ then adaptively measure $M_2$ according to; $\mathcal{O}(M_1) = 0 \longrightarrow M_2 = Z\mathbb{I}$, $\mathcal{O}(M_1) = 1 \longrightarrow M_2 = XX$.}.

\end{document}